\documentclass[11pt,letterpaper,english]{article}

\usepackage{amsmath,amssymb,amsthm}
\usepackage[ruled,vlined]{algorithm2e}
\usepackage{pgf, tikz}
\usepackage{xspace, units}
\usepackage{typearea}
\allowdisplaybreaks

\usepackage{lmodern}
\usepackage[T1]{fontenc}
\usepackage{textcomp}

\typearea{16}
                             
\usepackage{color}
\definecolor{Darkblue}{rgb}{0,0,0.4}
\definecolor{Brown}{cmyk}{0,0.81,1.,0.60}
\definecolor{Purple}{cmyk}{0.45,0.86,0,0}
\usepackage[breaklinks]{hyperref}
\hypersetup{colorlinks=true,
            citebordercolor={.6 .6 .6},linkbordercolor={.6 .6 .6},
            citecolor=black,urlcolor=Darkblue,linkcolor=black,
            pdfpagemode=UseThumbs,pdfstartview=FitH}
\newcommand{\lref}[2][]{\hyperref[#2]{#1~\ref*{#2}}}

\newtheorem{theorem}{Theorem}
\newtheorem{proposition}[theorem]{Proposition}

\newtheorem{corollary}[theorem]{Corollary}
\newtheorem{fact}[theorem]{Fact}

\newtheorem{observation}[theorem]{Observation}

\newcommand{\NN}{\ensuremath{\mathbb{N}}}

\newcommand{\RR}{\ensuremath{\mathbb{R}}}
\renewcommand{\O}{{\mathcal O}}
\newcommand{\R}{{\mathcal R}}
\renewcommand{\L}{{\mathcal L}}

\renewcommand{\Pr}[1]{\mbox{\rm\bf Pr}\left[#1\right]}
\newcommand{\Ex}[1]{\mbox{\rm\bf E}\left[#1\right]}

\newcommand{\noise}{\ensuremath{N}}
\newcommand{\measure}{\ensuremath{W}}

\DeclareMathOperator*{\ALG}{ALG}
\DeclareMathOperator*{\polylog}{polylog}

\author{Thomas Kesselheim}

\title{A Constant-Factor Approximation for Wireless Capacity Maximization with Power Control in the SINR Model%
  \thanks{Department of Computer Science, RWTH Aachen University,
    Germany. \texttt{thomask@cs.rwth-aachen.de}. This work
    has been supported by the UMIC Research Centre, RWTH Aachen
    University.}}

\begin{document}

\maketitle

\begin{abstract}\noindent
In modern wireless networks, devices are able to set the power for each transmission carried out. Experimental but also theoretical results indicate that such power control can improve the network capacity significantly. We study this problem in the physical interference model using SINR constraints.

In the \emph{SINR capacity maximization problem}, we are given $n$ pairs of senders and receivers, located in a metric space (usually a so-called fading metric). The algorithm shall select a subset of these pairs and choose a power level for each of them with the objective of maximizing the number of simultaneous communications. This is, the selected pairs have to satisfy the SINR constraints with respect to the chosen powers.

We present the first algorithm achieving a constant-factor approximation in fading metrics. The best previous results depend on further network parameters such as the ratio of the maximum and the minimum distance between a sender and its receiver. Expressed only in terms of $n$, they are (trivial) $\Omega(n)$ approximations.

Our algorithm still achieves an $\O(\log n)$ approximation if we only assume to have a general metric space rather than a fading metric. Furthermore, by using standard techniques the algorithm can also be used in single-hop and multi-hop scheduling scenarios. Here, we also get $\polylog n$ approximations. 
\end{abstract}
\vfill

\thispagestyle{empty}
\setcounter{page}0
\pagebreak[4]

\section{Introduction}
Wireless communication networks become more and more important these days -- for example as a cellular network, as a wireless LAN, or as sensor networks. All these networks have in common simultaneous communications are not physically separated but interfere. Therefore one uses different time slots or different channels. A promising technology to better utilize time slots or channels is so-called power control enabling each device to set its transmission power for each transmission. So the fundamental questions arising are not only when to transmit on which channel but also at which power level. 

In order to model the aspects of power control appropriately, the standard model in the engineering community is the so-called physical interference model (see e.\,g., \cite{Zander1992,ElBatt2004,Brar2006}). Here, a transmission can successfully be decoded if the SINR constraint is fulfilled. This is the ratio between the desired signal strength and all other signal strengths plus ambient noise (Signal-to-Interference-plus-Noise Ratio, SINR) is above some threshold $\beta$. Over the past years, this model also attracted much interested in algorithmic research (see e.\,g. \cite{Moscibroda2006,Halldorsson2009,Fanghaenel2009,Halldorsson2009a,Andrews2009}) because also theoretically power control has an enormous impact on the capacity of a wireless network by effects not captured by graph-based models \cite{Moscibroda2006b}.

The fundamental optimization problem considered is the capacity maximization problem. Here, one is given a set of $n$ pairs of senders and receivers located in a metric space. The objective is to find a maximum set of these requests and a power assignment such that the SINR constraint is fulfilled. Previous algorithms typically use transmission powers chosen as a function of the distance between the sender and its receiver. The best result in this area \cite{Halldorsson2010} achieves an approximation factor of $\O(\log \log \Delta + \log n)$. Here $\Delta$ is the ratio between the largest and the smallest distance between a sender and its receiver. However, this bound is tight and indeed its performance degenerates to $\Omega(n)$ for large $\Delta$. Since such a performance is also achieved by trivial algorithms it has been one of the main open problems to find algorithms with a non-trivial worst-case guarantee in terms of $n$ (see, e.\,g., \cite{Moscibroda2006,Andrews2009,Halldorsson2009a}). In this paper, we solve this question by presenting a surprisingly simple algorithm that achieves constant factor approximation in the case of a so-called fading metric. For general metrics, the same algorithm achieves an $\O(\log n)$ approximation.

\subsection{Formal Description of the Capacity Maximization Problem}    
Formally, the physical interference model is defined as follows. As stated above, the set of network nodes $V$ forms a metric space. This allows modelling the propagation of a radio signal as follows. Suppose a node $s$ transmits a signal at a power level of $p$, then some node $r$ receives it at a signal strength of $p / d(s, r)^\alpha$, where $d(s,r)$ is the distance between the two nodes. The constant $\alpha$ is called \emph{path-loss exponent}\footnote{Typically, it is assumed that $2 < \alpha < 6$. Our algorithms work for any $\alpha > 0$. However, the approximation factors improve when restricting $\alpha$.}. Such a transmission is successful if the strength of the desired signal is at least a factor $\beta$ higher than the strength of simultaneous transmissions plus ambient noise. This is the \emph{Signal-to-Interference-plus-Noise Ratio} is at least above some threshold $\beta > 0$, which is called \emph{gain}.
More formally, a set $\L \subseteq V \times V$ of pairs of senders and receivers (links) can communicate simultaneously on the same channel using the power assignment $p\colon \L \to \RR_{>0}$, if for all $\ell \in \L$ the SINR constraint
\[
\frac{p(\ell)}{d(s, r)^\alpha} \geq \beta \sum_{\ell' = (s', r') \in \L} \frac{p(\ell')}{d(s', r)^\alpha} + \noise
\]
is fulfilled. The constant $\noise \geq 0$ allows us to model ambient noise that all transmissions have to cope with. 

In fact, we consider a slight generalization of the capacity maximization problem. In the \emph{$k$ channel capacity maximization problem}, we are given a request set $\R \subseteq V \times V$ consisting of links. The objective is to find $k$ disjoint sets $\L_1, \ldots, \L_k \subseteq \R$ (corresponding to different time slots or non-interfering channels) of joint maximum cardinality and a power assignment such that the SINR constraint is fulfilled for each set $\L_t$. In contrast, in the \emph{interference scheduling problem} the objective is not to schedule as many requests as possible, but all of them minimizing $k$, the number of sets used.

One of the main challenges of both of these NP-hard problems is they are not convex. In particular, we have to select requests such that we get admissible sets, i.\,e. sets for which there is a correct power assignment. Deciding whether a set is admissible and finding the right power assignment can be done in polynomial time by solving a linear program. Nevertheless, the difficulty is still to find the admissible sets. There is no known way to use the LP approach in this context. Our approach is now to select the sets based on a sufficient condition. This condition is well chosen such that any admissible set has to fulfill a similar condition. Therefore, we can show we do not lose too many requests in this selection. 

We assume that both the optimum and also the calculated solution can select an arbitrary transmission power for each link. This common assumption is reasonable at least in theory because if for example the ratio between the maximum and the minimum transmission power is constant, the effect of power control on the solution quality is also only constant. However, of course, the case of limited transmission powers or limited overall energy could still be interesting for future work.

Another common assumption in theory and practice is the fact that nodes are located in the plane and that $\alpha > 2$. The advantage is the summed up interference by infinitely many signals sent out from equidistant senders at constant powers converges in this case. Halld\'orsson \cite{Halldorsson2009a} generalized this observation to so-called \emph{fading metrics}. Here, the doubling constant of the metric has to be strictly less than $\alpha$. This assumption often enables to build simpler algorithms or to get better performance guarantees. In particular, our algorithms do not depend on the assumption of a fading metric to work. However, we can prove better approximation ratios in this case.

\subsection{Our Results}
We present the first constant-factor approximation algorithm for the $k$ channel capacity maximization problem in fading metrics, particularly for the standard case of the Euclidean plane with $\alpha > 2$. Unlike previous theoretical approaches, the algorithm first selects the links and chooses transmission powers afterwards. The scheduling part is a simple greedy algorithm ensuring that a sufficient condition is met allowing the power control part to work. In order to analyze the quality of the solution, we find out that any admissible set has to fulfill a condition similar to the one of our greedy algorithm. This yields the mentioned approximation factor. 

For the case of arbitrary metrics the same algorithm still achieves a performance ratio of $\O(\log n)$. By applying it repeatedly, it can also be used to solve the interference scheduling problem. In this case, the approximation factors are $\O(\log n)$ in fading metrics resp.\ $\O(\log^2 n)$ in general metrics. Furthermore, the achieved necessary and sufficient conditions are linear allowing us to apply the algorithm in multi-hop variants. By using the technique of random delays, we get an $\O(\log^2 n)$ approximation for the multi-hop case with fixed paths. The optimization problem of finding optimal paths can be formulated as an integer linear program. By applying LP relaxation and randomized rounding, we achieve a total approximation factor of $\O(\log^2 n)$ resp.\ $\O(\log^3 n)$. Also for these scheduling problems, we achieve approximation factors that are significantly better than the ones known before. 

\subsection{Related Work}
\label{sec:related-work}
There has been lots of progress on scheduling and power control in the physical interference model in the last years. On the one hand, many heuristics have been presented to either solve the power control or the scheduling part or both at a time (see e.\,g., \cite{Zander1992,ElBatt2004,Brar2006}). A typical approach for power control is a fixed point iteration where in each step the power is reduced to the minimum one necessary to still satisfy the SINR constraint. This iteration converges as long as there is a feasible power assignment. To combine the approach with the scheduling part, one uses sufficient conditions that imply that a fixed point exist. However, these approaches lack of a polynomial running time and/or a guaranteed approximation factor. Moscibroda et al.~\cite{Moscibroda2007} discussed these issues and proved a bad worst-case performance for a number of heuristics.


After Moscibroda and Wattenhofer \cite{Moscibroda2006} had given the first worst-case results for a related problem (not dealing with arbitrary request sets), the theoretical analysis of the capacity maximization and the interference scheduling problem was considered by a number of authors independently \cite{Goussevskaia2007,Andrews2009,Chafekar2007,Fanghaenel2009}. However, the first approaches featuring a worst-case approximation guarantee work in quite a similar way. The request set $\R$ is divided into $\log \Delta$ length classes, where $\Delta$ is the ratio between the length of the longest and the shortest link in $\R$. In such a length class, the problem becomes significantly easier because one can use uniform transmission powers for all links. For example, Andrews and Dinitz \cite{Andrews2009} achieve an $\O(\log \Delta)$ approximation for the capacity maximization problem with $k=1$ in the plane. They furthermore prove the problem to be NP hard. Chafekar et al.~\cite{Chafekar2007} use similar ideas to give an approximation algorithm for the combined multi-hop scheduling and routing problem. Here, transmission powers $p(\ell)$ are chosen proportional to $d(\ell)^\alpha$ (linear power assignment) to get an $\O(\log^2 \Delta \cdot \log^2 \Gamma \cdot \log^2 n)$ approximation algorithm for the plane, where $\Gamma$ is the ratio of the maximum and minimum power used by the optimum. This was improved to $\O(\log \Delta \cdot \log^2 n)$ by Fangh\"anel et al.~\cite{Fanghaenel2009} and also extended to $\O(\log \Delta \cdot \log^3 n)$ in general metrics (without fading metrics assumption).

Nevertheless, already Moscibroda and Wattenhofer \cite{Moscibroda2006} show the performance of both uniform and linear power assignments can be an $\Omega(\log \Delta)$ factor worse than the optimal power assignment.  Square-root power assignments have a significantly better performance by setting the transmission power $p(\ell)$ proportional to $\sqrt{ d(\ell)^\alpha }$. They have been introduced by Fangh\"anel et al.~\cite{Fanghaenel2009a}. Halld\'orsson \cite{Halldorsson2009a} proved the optimal schedule length is at most an $\O(\log \log \Delta \cdot \log n)$ factor worse than the optimal power assignment in fading metrics. 

However, Fangh\"anel et al.~\cite{Fanghaenel2009a} not only prove that an $\Omega(\log \log \Delta)$ bound holds for square-root power assignments but also one cannot achieve reasonable approximation ratios in general by choosing transmission powers only based on the distance between the sender and its receiver. In particular, for each such oblivious power assignment there is an instance in which the maximum set fixed to this power assignment is an $\Omega(n)$ factor smaller than the optimal solution. In other words, there is an instance (where $\Delta$ is large) such that we only get a trivial approximation. Interestingly, this counterexample does not need a complex metric but nodes are located on a line. 

For problem variants, previous results are significantly better, e.\,g., when the transmission powers are also fixed for the optimum. For the case that all transmission powers are the same, Goussevskaia et al.~\cite{Goussevskaia2009} give a constant factor approximation for capacity maximization and this way an $\O(\log n)$ approximation for scheduling. For the more general case where transmission powers only have to follow certain monotonicity constraints, Kesselheim and V\"ocking \cite{Kesselheim2010} present a distributed protocol achieving an $\O(\log^2 n)$ approximation for the interference scheduling problem. Another example to be mentioned is the case of bidirectional communication. Here, when restricting the optimum to use symmetric transmission powers for both directions, one can achieve $\polylog n$ approximations using square-root power assignments \cite{Fanghaenel2009a,Halldorsson2009a}.
\section{Admissible Sets in Fading Metrics}
Before coming to the algorithm, let us first focus on necessary conditions an admissible set has to fulfill. In contrast to typical graph-based interference models, due to power control a receiver may be located closer to arbitrarily many different senders than the one whose signal it intends to receive. However, one can easily see this requires the distance between these other senders and their respective receiver to be still smaller than the distance to the first receiver. On the other hand, there can only be constantly many senders of larger links within this range. In this section, we prove a generalization of this necessary condition on the distances between nodes in admissible sets. In particular, we prove the following result.
\begin{theorem}
\label{theorem:lowerbound}
Let $\L$ be an admissible set. Let furthermore $s$ and $r$ be some arbitrary nodes in $V$. Then we have
\[
\sum_{\substack{(s', r') \in \L \\ d(s', r') \geq d(s, r)}} \min \left\{ 1, \frac{d(s, r)^\alpha}{d(s, r')^\alpha} \right\} + \min \left\{ 1, \frac{d(s, r)^\alpha}{d(s', r)^\alpha} \right\} = \O(1)
\]
in \emph{fading metrics}.
\end{theorem}
That is, we place some additional sender-receiver pair in an admissible set and and consider how many links having at least the same distance between the sender and its receiver can be at some distance. 

The sum consists of two terms for each link in $\L$, one that only depends on the sender and one that only depends on the receiver. In order to prove the theorem, we split up each link and -- in some sense -- we treat both endpoints separately. For a link set $\L$, we denote by $V(\L)$ the set of all endpoints of links in $\L$. In order to simplify notation, we assume no two links share an endpoint. For a node $v$ we denote by $\ell_v$ the link whose endpoint is $v$.


\begin{proposition}
\label{prop:fadingmetric}
Let $w \in V$ be an arbitrary node and $d>0$. Let $\L'$ be an admissible set of links of length at least $d$. Then we have
\[
\sum_{v \in V(\L')} \min \left\{ 1, \frac{d^\alpha}{d(v, w)^\alpha} \right\} = \O(1)
\]
in fading metrics.
\end{proposition}

Our proof is built upon two facts. On the one hand, we use a standard technique to decompose the link set $\L'$ into constantly many sets in which all sender and receiver nodes have a distance of at least $d$ to each other. On the other hand, we use the fading-metric property. Informally spoken, we regard all nodes as senders transmitting at the same power and having a minimum distance of $d$. Then in fading metrics the total interference by arbitrarily many senders converges. This is exactly the bound we need to prove the theorem.

\begin{proof}[Proof of Proposition~\ref{prop:fadingmetric}]
We use the technique of signal strengthening by Halld\'orsson and Wattenhofer \cite{Halldorsson2009}. This is, we decompose the link set $\L'$ to $\lceil 2 \cdot 3^{\alpha} / \beta \rceil^2$ sets, each admissible with a gain $\beta' = 3^\alpha$. We prove the claim for each of the (constantly many) sets separately. Consider two links $\ell =(u, v) \neq \ell' = (u', v')$ with $d(u, v') \leq d(u', v)$ in such a set $\L''$.

We know the SINR condition is fulfilled for some arbitrary power assignment with a gain of $\beta' = 3^\alpha$. So, we have
\[
\frac{p(\ell)}{d(s, r)^\alpha} \geq 3^\alpha \frac{p(\ell')}{d(s, r')^\alpha} \enspace, \qquad \text{ and} \qquad \frac{p(\ell')}{d(s', r')^\alpha} \geq 3^\alpha \frac{p(\ell)}{d(s', r)^\alpha} \enspace.
\]
By multiplying these two inequalities, we can conclude
\[
d(s, r') \cdot d(s', r) \geq 9 \cdot d(s, r) \cdot d(s', r') \enspace.
\]
On the other hand, by triangle inequality, we get
\[
d(s, r') \cdot d(s', r) \leq d(s, r') \left( d(s, r') + d(s', r') + d(s, r) \right) \enspace.
\]
We observe that in order to fulfill both conditions it is not possible that both $d(s, r') < 2 d(s', r')$ and $d(s, r') < 2 d(s, r)$ at the same time. Therefore $d(s, r') \geq \min\{2d(s, r), 2d(s', r')\}$. Using the fact that $d(s, r)$ and $d(s', r')$ are at least $d$ this means the distance between any of the involved nodes $s, r, s', r'$ has to be at least $d$.

Now we apply the fact we have a fading metric. Let $A < \alpha$ be the doubling dimension of the metric. This is, there is some absolute constant such that for all $x \in V$ and all $g>0$ any packing of balls of radius $Z$ inside a ball of radius $tZ$ contains at most $Ct^A$ balls. For $g>0$, let $T_g = \{v \in V(\L'') \mid d(v, w) < gd/2 \}$. We have seen above that the distance between any two nodes in $V(\L'')$ is at least $d$. This is balls of radius $d/2$ centered at the nodes in $T_g$ do not intersect and are fully contained in $B(w, (g+1)d/2)$. The packing constraint implies $\lvert T_g \rvert \leq C(g+1)^A$.

Therefore, we get
\[
\sum_{v \in V(\L'')} \min \left\{ 1, \frac{d^\alpha}{d(v, w)^\alpha} \right\} \leq \lvert T_2 \rvert + \sum_{g=2}^\infty \lvert T_{g+1} \setminus T_g \rvert \frac{d^\alpha}{(gd/2)^\alpha} \leq C \cdot 3^A + \sum_{g=2}^\infty \lvert T_g \rvert \left( \frac{2^\alpha}{(g-1)^\alpha} - \frac{2^\alpha}{g^\alpha} \right)
\]
For $g \geq 2$, we have (c.\,f. \cite{Halldorsson2009a})
\[
\frac{1}{(g-1)^\alpha} - \frac{1}{g^\alpha} = \frac{g^\alpha - (g-1)^\alpha}{(g-1)^\alpha g^\alpha}  \leq \frac{\alpha g^{\alpha - 1}}{g^\alpha (g-1)^\alpha}  \leq \frac{\alpha}{(g-1)^{\alpha + 1}} \leq \frac{3^{\alpha + 1} \alpha}{(g+1)^{\alpha + 1}} \enspace. 
\]
Using this result and the bound on $\lvert T_g \rvert$, we can bound the sum by
\[
C \cdot 3^A + \sum_{g=2}^\infty \frac{C 2^\alpha 3^{\alpha + 1} \alpha}{(g+1)^{\alpha - A + 1}} \leq C \cdot 3^A + \frac{C 2^\alpha 3^{\alpha + 1} \alpha}{\alpha - A} = \O(1) \enspace.
\]
This proves the claim.
\end{proof}

\section{Selection and Power Control Algorithm}
In the previous section, we found a necessary condition that all admissible sets have to fulfill, only based on the distances between the network nodes. Unfortunately, this condition is not sufficient for a set to be feasible. Nevertheless, we will show in this section a similar condition to be sufficient. Our algorithm for the $k$ channel capacity maximization problem then uses this condition to first select admissible sets $\L_t$. Afterwards we assign the transmission powers recursively in this set. The latter step of assigning the transmission powers could of course be exchanged by any other algorithm that chooses transmission powers for an admissible set such as the previously mentioned LP. However, also in this case our algorithm can be seen as a constructive proof the calculated set is indeed admissible.
 
In particular, our algorithm works as follows. Given the request set $\R$, consisting of pairs of senders and receivers that can be assigned to a set. We go through $\R$ in order of increasing length. For notational simplicity we assume that this ordering is unique by breaking ties arbitrarily, i.\,e. $d(\ell) \neq d(\ell')$ for all $\ell \neq \ell'$. Going through $\R$, we assign each link $\ell' = (s', r')$ to any set $\L_t$ for which the following condition is fulfilled:
\begin{equation}
\sum_{\substack{(s, r) \in \L_t \\ d(s, r) < d(s', r')}} \frac{d(s, r)^\alpha}{d(s, r')^\alpha} + \frac{d(s, r)^\alpha}{d(s', r)^\alpha} \leq \tau \qquad \text{where } \tau = \frac{1}{2 \cdot 3^\alpha \cdot \left(4 \beta + 2\right)} \enspace.
\label{eq:powercontrol-suffcond}
\end{equation}
If there is no such set $\L_t$, the link is discarded.

Afterwards, we select transmission powers as follows. Given an output set $\L_t$ of the above greedy selection, we go again through the links. This time in order of decreasing length setting the transmission power for a link $\ell' = (s', r')$ recursively as follows. We assign the longest link an arbitrary power (e.\,g. $1$). The further powers are set by
\[
p(\ell') = 4 \beta \sum_{\substack{(s, r) \in \L_t \\ d(s, r) > d(s', r')}} \frac{p(\ell)}{d(s, r')^\alpha} d(s', r')^\alpha \enspace.
\]
This is, we choose a transmission power that is proportional to the minimum transmission power needed to only deal with the longer links. Afterwards we scale the powers such that $p(\ell') \geq 2 \beta \noise \cdot d(\ell')$ for all $\ell'$.

\subsection{Feasibility of the Solution}
Let us first prove the algorithm calculates indeed a feasible solution. This is, for each set $\L_t$ the SINR constraint is fulfilled. We can see easily that neither interference from longer links nor noise can be significantly too large. Therefore the main part of our argumentation focuses on the interference from smaller links. Nonetheless, the smaller links themselves only use a transmission power proportional to the one necessary to compensate the larger ones. Using Condition~\ref{eq:powercontrol-suffcond} and the triangle inequality, we will see that their contribution is also not too much. 

\begin{theorem}
\label{theorem:powercontrol}
If in a link set $\L_t$ Condition~\ref{eq:powercontrol-suffcond} is fulfilled for all $\ell' \in \L_t$ then then the above procedure computes a power assignment fulfilling the SINR condition.
\end{theorem}

\begin{proof}
Let $\L_t = \{ (s_1, r_1), \ldots, (s_{\bar{n}}, r_{\bar{n}}) \}$ with $d(s_1, r_1) \geq d(s_2, r_2) \geq \ldots \geq d(s_{\bar{n}}, r_{\bar{n}})$. In this notation, for each link $(s_i, r_i)$ with $i > 1$ the power is set to
\[
p_i = 4 \beta \cdot \sum_{j=1}^{i-1} p_{i, j} \enspace, \qquad \text{where } \qquad p_{i,j}=\frac{p_j d(s_i, r_i)^\alpha}{d(s_j, r_i)^\alpha} \enspace.
\]

So the contribution of each link to the interference can again be divided into the contribution due to the different $p_{i, j}$ terms. Each of these terms can be thought of as the indirect effect of the previously chosen power $p_j$. The most important observation is now that we can bound the indirect effect of a link $k$ (via the adaptation of links that are smaller than link $k$) by its direct contribution to the interference as follows.
   
\begin{observation}
For $i, k \in [{\bar{n}}] := \{1, \ldots, \bar{n}\}$, we have
\[
\sum_{j=\max\{i, k\} + 1}^{\bar{n}} \frac{p_{j,k}}{d(s_j, r_i)^\alpha} \leq 2 \cdot 3^\alpha \cdot \tau \frac{p_k}{d(s_k, r_i)^\alpha} \enspace.
\]
\end{observation}

\begin{proof}
Let $m = \max\{i, k\} + 1$. We have
\[
\sum_{j=m}^{\bar{n}} \frac{p_{j,k}}{d(s_j, r_i)^\alpha} = \sum_{j=m}^{\bar{n}} \frac{p_k \cdot d(s_j, r_j)^\alpha}{d(s_j, r_i)^\alpha \cdot d(s_k, r_j)^\alpha} \enspace. 
\]

We split up the terms into two parts, namely $M_1 = \{ j \in [{\bar{n}}] \mid j \geq m, d(s_k, r_i) \leq 3 d(s_k, r_j) \}$, $M_2 = \{ j \in [{\bar{n}}] \mid j \geq m, d(s_k, r_i) > 3 d(s_k, r_j) \}$.

For all $j \in M_1$, we have $d(s_k, r_j) \geq \nicefrac{1}{3} \cdot d(s_k, r_i)$. This yields
\[
\sum_{j \in M_1} \frac{p_k \cdot d(s_j, r_j)^\alpha}{d(s_j, r_i)^\alpha \cdot d(s_k, r_j)^\alpha} \leq \frac{3^\alpha \cdot p_k}{d(s_k, r_i)^\alpha} \sum_{j \in M_1} \frac{d(s_j, r_j)^\alpha}{d(s_j, r_i)^\alpha} \leq \frac{3^\alpha \cdot p_k}{d(s_k, r_i)^\alpha} \cdot \tau
\]

For all $j \in M_2$, we have by triangle inequality
\[
d(s_k, r_i) \leq d(s_k, r_j) + d(s_j, r_j) + d(s_j, r_i) \leq \nicefrac{1}{3} \cdot d(s_k, r_i) + 2 d(s_j, r_i) \enspace,
\]
where the last step is due to the definition of $M_2$ and the fact that $d(s_j, r_j) \leq d(s_j, r_i)$ because of Condition~\ref{eq:powercontrol-suffcond}. This implies $d(s_j, r_i) \geq \nicefrac{1}{3} \cdot d(s_k, r_i)$ yielding
\[
\sum_{j \in M_2} \frac{p_k \cdot d(s_j, r_j)^\alpha}{d(s_j, r_i)^\alpha \cdot d(s_k, r_j)^\alpha} \leq \frac{3^\alpha p_k}{d(s_k, r_i)^\alpha} \sum_{j \in M_2} \frac{d(s_j, r_j)^\alpha}{d(s_k, r_j)^\alpha} \leq \frac{3^\alpha \cdot p_k}{d(s_k, r_i)^\alpha} \cdot \tau 
\]

Altogether this yields the claim.
\end{proof}

Let us now consider the interference a particular link $i$ is exposed to. Let the contribution by larger links be denote by $I_>$ and the one by smaller links by $I_<$. By definition of the powers for $I_<$, we get
\[
I_< = \sum_{j=i+1}^{\bar{n}} \frac{p_j}{d(s_j, r_i)^\alpha} = \sum_{j=i+1}^{\bar{n}} 4 \beta \sum_{k=1}^{j-1} \frac{p_{j, k}}{d(s_j, r_i)^\alpha} \enspace.
\]
By rearranging the sum this is
\[
4 \beta \sum_{j=i+1}^{\bar{n}} \sum_{k=1}^{i} \frac{p_{j, k}}{d(s_j, r_i)^\alpha} + 4 \beta \sum_{j=i+1}^{\bar{n}} \sum_{k=i+1}^{j-1} \frac{p_{j, k}}{d(s_j, r_i)^\alpha}
= 4 \beta \sum_{k=1}^i \sum_{j=i+1}^{\bar{n}} \frac{p_{j,k}}{d(s_j, r_i)^\alpha} + 4 \beta \sum_{k=i+1}^{\bar{n}} \sum_{j=k+1}^{\bar{n}} \frac{p_{j,k}}{d(s_j, r_i)^\alpha} \enspace.
\]
Using the above observation, this is at most
\[
2 \cdot 3^\alpha \cdot \tau 4 \beta \sum_{k=1}^{i-1} \frac{p_k}{d(s_k, r_i)^\alpha} + 2 \cdot 3^\alpha \cdot \tau 4 \beta \frac{p_i}{d(s_i, r_i)^\alpha} + 2 \cdot 3^\alpha \cdot \tau 4 \beta \sum_{k=i+1}^{\bar{n}} \frac{p_k}{d(s_k, r_i)^\alpha}
\]
In the first part, we can recognize the definition of $p_i / d(s_i, r_i)^\alpha$ and in the third part the one of $I_<$. Re-substituting both, we get
\[
I_< \leq \left(1+4 \beta\right) \cdot 2 \cdot 3^\alpha \cdot \tau \frac{p_i}{d(s_i, r_i)^\alpha} + 2 \cdot 3^\alpha \cdot \tau 4 \beta \cdot I_< \enspace.
\]

This implies
\[
I_< \leq \frac{\left(1+4 \beta \right) \cdot 2 \cdot 3^\alpha \cdot \tau}{1- 2 \cdot 3^\alpha \cdot \tau 4 \beta} \frac{p_i}{d(s_i, r_i)^\alpha} = \frac{1}{4 \beta} \frac{p_i}{d(s_i, r_i)^\alpha} \enspace. 
\]
By definition of the powers, we have $p_i / d(s_i, r_i)^\alpha \geq 4 \beta I_>$ and $p_i / d(s_i, r_i)^\alpha \geq 2 \beta \noise$. In total, we get for interference plus noise
\[
\sum_{j \neq i} \frac{p_j}{d(s_j, r_i)^\alpha} + \noise = I_< + I_> + \noise \leq \frac{1}{\beta} \frac{p_i}{d(s_i, r_i)^\alpha} \enspace.  
\]
This is the SINR constraint is fulfilled.
\end{proof}

In total our simple greedy algorithm computes a feasible approximate solution for the $k$ channel capacity maximization problem. Using the similarity between the sufficient (greedy) condition and the necessary condition from Theorem~\ref{theorem:lowerbound} will allow us derive the approximation factor in the next section.
\section{Deriving the Approximation Factor}
In order to simplify notation in the comparison of our greedy solution to an optimal one, we define an edge-weighted graph as follows. Each vertex represents a link in $\R$. For each pair of links $\ell = (s, r)$, $\ell' = (s', r')$, we define the edge weight by
\begin{equation}
\label{eq:edgeweights}
w(\ell, \ell') = \begin{cases}
\min \left\{ 1, \frac{d(s, r)^\alpha}{d(s, r')^\alpha} \right\} + \min \left\{ 1, \frac{d(s, r)^\alpha}{d(s', r)^\alpha} \right\} & \text{if $d(\ell') < d(\ell)$} \\
0 & \text{otherwise}
\end{cases} \enspace. 
\end{equation}
So this weight $w$ expresses the impact of $\ell$ on $\ell'$ in the greedy algorithm. In particular, Condition~\ref{eq:powercontrol-suffcond} of the greedy algorithm is equivalent to $\sum_{\ell \in \L} w(\ell, \ell') \leq \tau$.

Furthermore, for a link set $\L \subseteq \R$, we set
\[
\measure_{\ell}(\L) = \sum_{\ell' \in \L} w(\ell, \ell') \qquad \text{ and } \qquad \measure(\L) = \max_{\ell \in \R} \measure_{\ell}(\L) \enspace.
\]
So $\measure_\ell(\L)$ is the sum of edge-weights going out from $\ell$ into the set $\L$ and $\measure(\L)$ the maximum over all $\ell \in \R$. Note that $\ell$ may be any link from $\R$ and does not necessarily belong to $\L$. Theorem~\ref{theorem:lowerbound} shows that for all admissible sets $\measure(\L) = \O(1)$ in fading metrics.

\begin{theorem}
\label{theorem:greedyapproximationfactor}
Let $\L \subseteq \R$ be an arbitrary subset of the request set. Then the greedy algorithm calculates sets of joint cardinality $k \tau \big/ ( \measure(\L) + k \tau ) \cdot \lvert \L \rvert$.
\end{theorem}

\begin{proof}
Let $\L_1, \ldots, \L_k$ be the sets calculated by the algorithm and $\ALG$ be their union.

By definition, we have for all $\ell \in \R$
\[
\sum_{\ell' \in \L \setminus \ALG} w(\ell, \ell') \leq \measure(\L) \enspace.
\]
This yields
\[
\sum_{\ell \in \ALG} \sum_{\ell' \in \L \setminus \ALG} w(\ell, \ell') \leq \measure(\L) \cdot \lvert \ALG \rvert \enspace.
\]

On the other hand, all $\ell' \in \L \setminus \ALG$ are discarded by the algorithm because Condition~\lref{eq:powercontrol-suffcond} is violated for all sets $\L_t$ calculated by the algorithm. Summing up all resulting inequalities, we get
\[
\sum_{\ell \in \ALG} w(\ell, \ell') = \sum_{t=1}^k \sum_{\ell \in \L_t} w(\ell, \ell') > k \tau 
\]
This yields
\[
\sum_{\ell' \in \L \setminus \ALG} \sum_{\ell \in \ALG} w(\ell, \ell') > k \tau \cdot \lvert \L \setminus \ALG \rvert 
\]

In combination this is
\[
k \tau \cdot \lvert \L \setminus \ALG \rvert < \measure(\L) \cdot \lvert \ALG \rvert
\]
Therefore
\[
k \tau \cdot \lvert \L \rvert < \left( \measure(\L) + k \tau \right) \cdot \lvert \ALG \rvert
\]
\end{proof}

Choosing $\L$ as the union of the $k$ admissible sets in the optimal solution, we get the approximation factor for the greedy algorithm.
\begin{corollary}
The greedy algorithm approximates the $k$ channel capacity maximization problem within a factor of $\O(1)$ in fading metrics.
\end{corollary}

A further aspect to mention is the feasibility was only stated for unidirectional communication, whereas the typical goal is establishing bidirectional communication. Nevertheless, our algorithm is robust against this issue. In the proof of Theorem~\ref{theorem:lowerbound}, we treated senders and receivers the same way. Therefore, reversing links in the optimum does not change the result and we can also apply the algorithm for half duplex and full duplex. This is a much stronger result than on previous bidirectional models \cite{Fanghaenel2009a,Halldorsson2009a}, where both directions of a bidirectional link have to use the same transmission powers.

\section{General Metrics}
Taking the algorithm again into consideration, we observe that for feasibility it only needs the triangle inequality to hold. So, it computes correct solutions even in the case the fading metric assumption does not hold. However, Theorem~\ref{theorem:lowerbound} made use of the fading-metric property and indeed it does not hold in general metrics. However, in this case, we can still prove the following.
\begin{theorem}
\label{theorem:admissibleset-generalmetric}
Let $\L$ be an admissible set. Let furthermore $s$ and $r$ be some arbitrary nodes in $V$. Then we have
\[
\sum_{\substack{(s', r') \in \L \\ d(s', r') \geq d(s, r)}} \min \left\{ 1, \frac{d(s, r)^\alpha}{d(s, r')^\alpha} \right\} + \min \left\{ 1, \frac{d(s, r)^\alpha}{d(s', r)^\alpha} \right\} = \O(\log \lvert \L \rvert) 
\]
in any metric.
\end{theorem}
Although one has to distinguish some different cases, the proof works in a similar way to the one for Theorem~\ref{theorem:lowerbound}. It can be found in Appendix~\ref{sec:admissibleset-generalmetric-proof}.

This theorem directly yields that $\measure(\L) = \O(\log n)$ for admissible sets $\L$ in general metrics. Therefore our algorithm calculates an $\O(\log n)$ approximation in this case. Note that the bound in Theorem~\ref{theorem:admissibleset-generalmetric} is tight. So at least for our analysis it makes a difference if the fading metric assumption holds.
\section{Scheduling Single-Hop Requests and the Scheduling Complexity}
Apart from the capacity maximization problems, the most common problem is the interference scheduling problem. This time, we have to calculate a schedule for all links in $\R$ and aim at minimizing the number of time slots used. In order to solve the single-hop version of the interference scheduling problem, we use a slight variant of our capacity maximization algorithm. This time the number $k$ of sets we fill is not fixed. Therefore, when going through the links from small to large we now put each link into the first time slot for which Condition~\ref{eq:powercontrol-suffcond} holds.

\begin{theorem}
The greedy algorithm calculates a schedule of length $\O(\measure(\R) \cdot \log n)$.
\end{theorem}

\begin{proof}
Let $n_t$ be the number of links that are not scheduled within the first $t$ time slots. We now consider the reduction of $n_t$ in a single step. Essentially each of these steps is a run of the greedy capacity maximization algorithm on the remaining requests with $k=1$. In Theorem~\ref{theorem:greedyapproximationfactor}, we proved that 
\[
n_t - n_{t+1} \geq \tau \big/ \left( \measure(\R) + \tau \right) n_t
\]
This yields
\[
n_{t + 1} \leq \left( 1 - \frac{1}{\measure(\R) \frac{1}{\tau} + 1} \right) n_t
\]
So for $t > ( \measure \frac{1}{\tau} + 1 ) \cdot \ln n$, we have $n_t < 1$. Therefore, we need at most $\O(\measure(\R) \cdot \log n)$ steps.
\end{proof}

Using Theorem~\ref{theorem:lowerbound}, we can compare this length easily to the optimal schedule length $T(\R)$. We have $\measure(\R) = \O(T(\R))$ for fading metrics. In general metrics, we have $\measure(\R) = \O(T(\R) \cdot \log n)$. So the algorithm computes an $\O(\log n)$ resp. $\O(\log^2 n)$ approximation.

Furthermore, we can give a much more precise answer to an open question by Moscibroda et al.~\cite{Moscibroda2006a} than previous results: ``What is the time required to physically schedule an arbitrary set of requests?'' Up to an $\O(\log n)$ factor this \emph{scheduling complexity} $T(\R)$ is simply $\measure(\R)$, this is maximum weighted out-degree in our edge-weighted graph.
\section{Scheduling Multi-Hop Requests}
In contrast to a cellular network, where transmission take place between a mobile device and a base station, sensor networks typically use multi-hop communications between devices of the same kind. So, the store-and-forward routing problems known from wired networks can be seen in a new context here. Interestingly existing approaches for wired networks work well together with our approach for wireless scheduling. In this section, we will show how to apply these techniques in our case to build approximation algorithms for the multi-hop scheduling problem where paths are fixed, e.\,g. by routing tables, and for the joint routing and scheduling problem, where paths may be chosen arbitrarily from an edge set.

\subsection{Scheduling Packets on Fixed Paths}
In the multi-hop scenario with fixed paths requests are not pairs of senders and receivers but sequences of network nodes (paths) that packets have to be sent along. We denote the paths by $P_1, \ldots, P_m$, where $P_{i, j}$ denotes the $i$th node on the $j$th path. In terms of links, the set $\R$ of all $\ell_{i, j} := (P_{i,j}, P_{i, j+1})$ has to be scheduled. The additional constraint is that $\ell_{i, j}$ has be served before $\ell_{i, j+1}$ for all $i$ and $j$. This is also an additional limitation of the optimal schedule length. We express it in the parameter $D$ (dilation), which is the length of the longest path. If $T$ is the optimal schedule length, we have, of course, $D \leq T$. Furthermore, the bounds $\measure(\R) = \O(T)$ resp. $\measure(\R) = \O(T \cdot \log n)$ still hold.

Our algorithm uses the technique of random delays that has been used for wired networks \cite{Leighton1994} but also for wireless networks \cite{Chafekar2007,Fanghaenel2009,Kesselheim2010}. It assigns each packet a delay $1 \leq \delta_i \leq \measure (\R) / 3 \ln n$. This way, we get $\measure(\R) / 3 \ln n + D$ single-hop problems. On each of them, we apply the single-hop algorithm and get the multi-hop schedule as a concatenation of these schedules. 

\begin{theorem}
The multi-hop algorithm calculates a schedule of length $\O(\measure(\R) \cdot \log n + D \cdot \log^2 n)$ whp.
\end{theorem}

\begin{proof}
Let $X_{i, j, t} = 1$ iff $\delta_i = t + j$. Now we consider the sets $\L_t = \{ \ell_{i, j} \in V \times V \mid t = \delta_i + j \}$, which are used as inputs for the single-hop algorithm. We have for all $\ell \in V \times V$
\[
\measure_\ell(\L_t) = \sum_{i, j} w(\ell, \ell_{i, j}) \cdot X_{i, j, t} \enspace.
\]
So, we have $\Ex{\measure_\ell(\L_t)} \leq 3 \ln n$. For each $t$ the random variables $X_{i, j, t}$ are \emph{negatively associated} \cite{Dubhashi1998}. Therefore, we can apply a Chernoff bound to get.
\[
\Pr{\measure_\ell(\L_t) \geq (1 + \kappa) 3 \ln n} \leq \exp( - \ln n \cdot \kappa)
\]
This is, with probability $1 - n^{-\kappa + 1}$, for no link $\ell_{i, j}$ the random variable $\measure_{\ell_{i,j}}(\L_t)$ is above $(1 + \kappa) 3 \ln n$ for the set $\L_t$ this link is contained in. Therefore, the single-hop algorithm needs $\O(\kappa \log n)$ time slots for this set $\L_t$. This yields the total schedule length does not exceed $\O(\kappa \measure(\R) \log n + \kappa D \log^2 n)$. This proves the claim.
\end{proof}
In comparison to this optimal schedule length, this is an $\O(\log^2 n)$ approximation. 

\subsection{Routing}
Now consider the problem that paths are no more fixed but we have to select the ones minimizing the total time until delivery from an edge set $E$. The joint problem of finding paths, powers and a schedule has been stated by Chafekar et al.~\cite{Chafekar2007} as \emph{cross-layer latency minimization problem} (CLM problem). Here, we are given $m$ source-destination pairs $(s_i, t_i)$ and we have to find paths, powers and a schedule such that the time until the delivery of the last packet is minimized.  The best algorithm for this problem up to now was an $\O(\log \Delta \cdot \log^2 n)$ approximation in the plane and an $\O(\log \Delta \cdot \log^3 n)$ approximation for general metrics by Fangh\"anel et al. \cite{Fanghaenel2009}. In order to get an approximation algorithm for the CLM problem using the above results the only missing piece is the choice of paths.

For path selection, both Chafekar et al. \cite{Chafekar2007} and Fangh\"anel et al. \cite{Fanghaenel2009} adapt an LP rounding approach, which has been successfully applied in wired networks \cite{Raghavan1988}. Only very few modifications are needed to adapt this standard approach in our case. In particular, we use the following LP formulation. Here $N_{\text{in}}(v)$ resp. $N_{\text{out}}(v)$ denote the incoming resp. outgoing edges from $v$.
\begin{subequations}
\label{eq:optimalpaths}
\begin{align}
& \min z \\
\text{s.\,t.} & \sum_{e \in N_{\text{out}}(s_i)} y(i, e) - \sum_{e \in N_{\text{in}}(s_i)} y(i, e) = 1 && \text{for all } i \in [m] \label{eq:optimalpaths:flow1}\\
& \sum_{e \in N_{\text{out}}(v)} y(i, e) - \sum_{e \in N_{\text{in}}(v)} y(i, e) = 0 && \text{for all } i \in [m], v \in V \backslash \{ s_i, t_i \} \label{eq:optimalpaths:flow2}\\
& \sum_{e \in E} y(i, e) \leq z &&  \text{for all } i \in [m] \label{eq:optimalpaths:dilation}\\
& \sum_{i \in [m]} \sum_{e'} w(\ell, e') \cdot y(i, e') \leq z && \text{for all } \ell \in V \times V \label{eq:optimalpaths:sinr}\\
& y(i, e) \in \{ 0, 1\} && \text{for all } i \in [m], e \in E \label{eq:optimalpaths:integer}
\end{align}
\end{subequations} 
In this LP, the variables $y(i, e)$ indicate if path $i$ contains edge $e$ at any point. Feasible choices of paths establish an integral multi-commodity flow, which is expressed by Constraints~\ref{eq:optimalpaths:flow1}, \ref{eq:optimalpaths:flow2}, and \ref{eq:optimalpaths:integer}.  The objective is to minimize $z = \max\{ \measure(\R), D \}$. Constraint~\ref{eq:optimalpaths:dilation} ensures that $D \leq z$, Constraint~\ref{eq:optimalpaths:sinr} that $\measure_\ell(\R) \leq z$ for all $\ell \in V \times V$.

The relaxation is derived by exchanging Constraint~\ref{eq:optimalpaths:integer} by $0 \leq y(i, e) \leq 1$. Given an optimal fractional solution, it is rounded to an integer one in three steps. First the multi-commodity flow is decomposed into its paths. Afterwards paths longer that $2z$ are removed and their flow is distributed among the other paths for the same commodity. In the last step the actual rounding takes place. For each commodity the path is chosen independently at random using the flow as a probability distribution among the paths. One can easily show that if $z^\ast$ is the value of the fractional solution, we get a collection of paths with $D \leq 2 z^\ast$ and $\measure(\R) = \O(z^\ast \cdot \log n)$ whp this way.  

In combination, we get the following guarantees for the CLM problem. Observe that on the one hand the multi-hop scheduling algorithm computes a schedule of length $\O(z^\ast \cdot \log^2 n)$ whp. On the other hand, consider the optimal solution resulting in a schedule of length $T$. The paths represent an LP solution. Using the fact that $D \leq T$, and $\measure(\R) = \Omega(T)$ resp. $\measure(\R) = \Omega(T / \log n)$, we have $z^\ast = \Omega(T)$ resp. $z^\ast = \Omega(T / \log n)$. So, the schedule calculated by our algorithm has length $\O(T \cdot \log^2 n)$ resp. $\O(T \cdot \log^3 n)$. This is we get an $\O(\log^2 n)$ approximation in fading metrics and an $\O(\log^3 n)$ approximation in general metrics.
\section{Conclusions and Open Problems}
We showed in this paper that the fundamental optimization problems in the physical interference model can be approximated quite well by rather simple algorithms. However, still a some questions and problems stay unsolved. For example, a number of $\log n$ factors remain that could potentially be removed. It is also not clear, how small the constant for capacity maximization in fading metrics can become or if there is even a PTAS. In general, there are surprisingly few hardness results for all of the problems.

Furthermore, we found an abstraction to apply well-known techniques from wired networks to solve further problems and get a significantly better approximation than any previous algorithm. Probably results for more complex problems from the wired world could also be transferred -- such as dynamic packet injection -- or also results from graph-based interference models -- such as topology control.

Apart from this, to improve the practical relevance, it would be important to de-centralize the algorithm. Indeed it has many parallels to distributed heuristics on the problem whose ideas could potentially be transferred. Another practical issue is we assumed that neither the optimum nor our solution have any limitations on the transmission powers. Probably a number of insights could be transferred to the case of restricted transmission powers.

\clearpage
\thispagestyle{empty}
\bibliographystyle{plain}
\bibliography{bibliography}

\begin{thebibliography}{10}

\bibitem{Andrews2009}
Matthew Andrews and Michael Dinitz.
\newblock Maximizing capacity in arbitrary wireless networks in the sinr model:
  Complexity and game theory.
\newblock In {\em Proceedings of the 28th Conference of the IEEE Communications
  Society (INFOCOM)}, 2009.

\bibitem{Brar2006}
Gurashish Brar, Douglas~M. Blough, and Paolo Santi.
\newblock Computationally efficient scheduling with the physical interference
  model for throughput improvement in wireless mesh networks.
\newblock In {\em Proceedings of the 12th annual international conference on
  Mobile computing and networking (MobiCom)}, pages 2--13. ACM, 2006.

\bibitem{Chafekar2007}
Deepti Chafekar, V.~S.~Anil Kumar, Madhav~V. Marathe, Srinivasan Parthasarathy,
  and Aravind Srinivasan.
\newblock Cross-layer latency minimization in wireless networks with {SINR}
  constraints.
\newblock In {\em Proceedings of the 8th ACM International Symposium Mobile
  Ad-Hoc Networking and Computing (MOBIHOC)}, pages 110--119, 2007.

\bibitem{Dubhashi1998}
Devdatt~P. Dubhashi and Desh Ranjan.
\newblock Balls and bins: A study in negative dependence.
\newblock {\em Random Structures and Algorithms}, 13(2):99--124, 1998.

\bibitem{ElBatt2004}
Tamer~A. ElBatt and Anthony Ephremides.
\newblock Joint scheduling and power control for wireless ad hoc networks.
\newblock {\em IEEE Transactions on Wireless Communication}, 3(1):74--85, 2004.

\bibitem{Fanghaenel2009a}
Alexander Fangh{\"a}nel, Thomas Kesselheim, Harald R{\"a}cke, and Berthold
  V{\"o}cking.
\newblock Oblivious interference scheduling.
\newblock In {\em Proceedings of the 28th ACM symposium on Principles of
  distributed computing}, pages 220--229, 2009.

\bibitem{Fanghaenel2009}
Alexander Fangh{\"a}nel, Thomas Kesselheim, and Berthold V{\"o}cking.
\newblock Improved algorithms for latency minimization in wireless networks.
\newblock In {\em Proceedings of the 36th International EATCS Colloquium on
  Automata, Languages and Programming (ICALP)}, pages 447--458, 2009.

\bibitem{Goussevskaia2009}
Olga Goussevskaia, Magn{\'u}s~M. Halld{\'o}rsson, Roger Wattenhofer, and Emo
  Welzl.
\newblock Capacity of arbitrary wireless networks.
\newblock In {\em Proceedings of the 28th Conference of the IEEE Communications
  Society (INFOCOM)}, 2009.

\bibitem{Goussevskaia2007}
Olga Goussevskaia, Yvonne~Anne Oswald, and Roger Wattenhofer.
\newblock Complexity in geometric {SINR}.
\newblock In {\em Proceedings of the 8th ACM International Symposium Mobile
  Ad-Hoc Networking and Computing (MOBIHOC)}, pages 100--109, New York, NY,
  USA, 2007.

\bibitem{Halldorsson2009a}
Magn{\'u}s~M. Halld{\'o}rsson.
\newblock Wireless scheduling with power control.
\newblock In {\em Proceedings of the 17th annual European Symposium on
  Algorithms (ESA)}, pages 361--372, 2009.

\bibitem{Halldorsson2010}
Magn{\'u}s~M. Halld{\'o}rsson and Pradipta Mitra.
\newblock Wireless capacity with oblivious power in general metrics.
\newblock unpublished manuscript, 2010.

\bibitem{Halldorsson2009}
Magn\'{u}s~M. Halld\'{o}rsson and Roger Wattenhofer.
\newblock Wireless communication is in {APX}.
\newblock In {\em Proceedings of the 36th International EATCS Colloquium on
  Automata, Languages and Programming (ICALP)}, pages 525--536, 2009.

\bibitem{Kesselheim2010}
Thomas Kesselheim and Berthold V{\"o}cking.
\newblock Distributed contention resolution in wireless networks.
\newblock In {\em Proceedings of the 24th International Symposium on
  Distributed Computing (DISC)}, 2010.
\newblock to appear, available at
  \url{http://www.algo.rwth-aachen.de/~thomask/}.

\bibitem{Leighton1994}
F.~T. Leighton, Bruce~M. Maggs, and Satish~B. Rao.
\newblock Packet routing and job-shop scheduling in {O(congestion+dilation)}
  steps.
\newblock {\em Combinatorica}, 1994.

\bibitem{Moscibroda2007}
Thomas Moscibroda, Yvonne~Anne Oswald, and Roger Wattenhofer.
\newblock How optimal are wireless scheduling protocols?
\newblock In {\em Proceedings of the 26th Conference of the IEEE Communications
  Society (INFOCOM)}, 2007.

\bibitem{Moscibroda2006}
Thomas Moscibroda and Roger Wattenhofer.
\newblock The complexity of connectivity in wireless networks.
\newblock In {\em Proceedings of the 25th Conference of the IEEE Communications
  Society (INFOCOM)}, pages 1--13, 2006.

\bibitem{Moscibroda2006b}
Thomas Moscibroda, Roger Wattenhofer, and Yves Weber.
\newblock Protocol design beyond graph-based models.
\newblock In {\em In Proceedings of the 5th ACM SIGCOMM Workshop on Hot Topics
  in Networks (HotNets) 2006.}, 2006.

\bibitem{Moscibroda2006a}
Thomas Moscibroda, Roger Wattenhofer, and Aaron Zollinger.
\newblock Topology control meets {SINR}: {T}he scheduling complexity of
  arbitrary topologies.
\newblock In {\em Proceedings of the 7th ACM International Symposium Mobile
  Ad-Hoc Networking and Computing (MOBIHOC)}, pages 310--321, 2006.

\bibitem{Raghavan1988}
Prabhakar Raghavan.
\newblock Probabilistic construction of deterministic algorithms: approximating
  packing integer programs.
\newblock {\em Journal of Computer and System Sciences}, 37(2):130--143, 1988.

\bibitem{Zander1992}
J.~Zander.
\newblock {Performance of optimum transmitter power control in cellular radio
  systems}.
\newblock {\em IEEE Transactions on Vehicular Technology}, 41(1):57, 1992.

\end{thebibliography}

\begin{appendix}
\section{Proof of Theorem~\ref{theorem:admissibleset-generalmetric}}
\label{sec:admissibleset-generalmetric-proof}
In order to prove Theorem~\ref{theorem:admissibleset-generalmetric}, we will split up the sum into two parts. We distinguish between \emph{near links} and \emph{far links}. The next section will only deal with the contribution near links. This is one of the endpoints is located closer to $w$ than to its actual counterpart. Afterwards, we will bound the contribution of the remaining links.

\subsection{Bounding the Contribution of Near Links}
Let us first bound the contribution of \emph{near links}, i.\,e., links for which one endpoint $v$ lies closer to $w$ than to its counterpart, $d(v, w) \leq d(\ell_v)$. As a matter of fact, there can be arbitrarily many such link. Nevertheless, we can prove that their distance to $w$ increases exponentially. Therefore, we can achieve a constant bound for any metric.
\begin{proposition}
Let $w \in V$ be an arbitrary node and $d>0$. Let $\L'$ be a admissible set of links of length at least $d$ such that for one of the of endpoints $v$ we have $d(v, w) < d(\ell_v)$. Then it holds
\[
\sum_{v \in V(\L)} \min \left\{ 1, \frac{d^\alpha}{d(v, w)^\alpha} \right\} = \O(1)
\]
in any metric.
\end{proposition}

\begin{proof}
Again, we use the technique of signal strengthening by Halld\'orsson and Wattenhofer \cite{Halldorsson2009}. We now decompose the link set $\L'$ to $\lceil 2 \cdot 4^{\alpha} / \beta^2 \rceil$ sets, each admissible with a gain $\beta' = 2^\alpha$. We prove the claim for each of the (constantly many) sets separately.

We divide such a set $\L''$ into further sets as follows. Let $T = \{ (s, r) \in \L \mid \min \{ d(s, w), d(r, w) \} < d \}$ be the set of links having one endpoint at distance less than $d$ from $w$. Furthermore, for $a \geq 0$ let $U_a = \{ (s, r) \in \L \mid 2^a d \leq \min \{ d(s, w), d(r, w) \} < 2^{a+1} d \}$. We claim each of the sets $T$ and $U_a$, $a>0$, contains at most one link.

Suppose any of the sets contains two links $\ell = (s, r) \neq \ell' = (s', r')$. Let $d(s, r) \geq d(s', r')$. Observe that due to the triangle inequality the smaller ones of the distances $d(s, r')$ and $d(s', r)$ is less than $4 d(s', r')$, the larger one is at most $4 d(s, r)$. 
 
We know the SINR condition is fulfilled for some arbitrary power assignment with a gain of $\beta' = 4^\alpha$. This is, we have
\[
\frac{p(\ell)}{d(s, r)^\alpha} \geq 4^\alpha \frac{p(\ell')}{d(s', r)^\alpha} \enspace, \qquad \text{ and } \qquad
\frac{p(\ell')}{d(s', r')^\alpha} \geq 4^\alpha \frac{p(\ell')}{d(s, r')^\alpha} \enspace.
\]
By multiplying these two inequalities, we can conclude
\[
d(s, r') \cdot d(s', r) \geq 16 \cdot d(s, r) \cdot d(s', r') \enspace.
\]
This is a contradiction to the above observation.

Having bounded the cardinality of the sets, we can use them to bound the sum as follows.
\[
\sum_{v \in V(\L'')} \min \left\{ 1, \frac{d^\alpha}{d(v, w)^\alpha} \right\} \leq 2 \lvert T \rvert + \sum_{a=0}^\infty 2 \lvert U_a \rvert \frac{d^\alpha}{(2^a d)^\alpha} \leq 2 + 2 \sum_{a=0}^\infty \frac{1}{(2^a)^\alpha} = 2 + \frac{2}{1 - 2^{-\alpha}} = \O(1) \enspace.
\]
\end{proof}

\subsection{Bounding the Contribution of Far Links}
Now consider the links having the property that for both endpoints the corresponding endpoint lies closer than $w$. We use a proof technique that has been used before in \cite{Fanghaenel2009}.

\begin{proposition}
\label{prop:generalmetric:farlinks}
Let $w \in V$ be an arbitrary node and $d>0$. Let $\L'$ be a admissible set of links of length at least $d$ such that for both endpoints $v$ we have $d(v, w) \geq d(\ell_v)$. Then it holds
\[
\sum_{v \in V(\L)} \min \left\{ 1, \frac{d^\alpha}{d(v, w)^\alpha} \right\} = \O(\log n)
\]
in any metric.
\end{proposition}

Before coming to the actual proof of the lemma, we need a bound on the number of nodes that can at most lie within a distance of $Z$ from $w$. 

\begin{fact}
\label{fact:numberofnodesinneighborhood}
For all $Z > 0$ we have for $K_{Z}(w) = \left\{ v \in V(\L') \mid d(v, w) \leq Z \right\}$
\[
\lvert K_{Z}(w) \rvert \leq \frac{2}{\beta} \left( \frac{4 Z}{d} \right)^\alpha + 2 \enspace.
\]
\end{fact}

\begin{proof}
Let $(u_1, v_1), \ldots, (u_j, v_j)$ be the links each of which has at least one endpoint in $K_Z(w)$. Note that the corresponding endpoints lies within a distance of $2Z$ from $w$. Let $p$ be the power assignment allowing all links to be admissible. Furthermore, w.\,l.\,o.\,g., let $(s_1, r_1)$ be the link with minimal power $p_1$. As the SINR condition is satisfied for link $(s_1, r_1)$, we get:
\[
\frac{1}{\beta} \frac{p_1}{d(s_1, r_1)^\alpha} \geq \sum_{i \neq 1} \frac{p_i}{d(s_i, r_1)^\alpha} \geq \sum_{i \neq 1} \frac{p_i}{(4 Z)^\alpha} \geq \frac{(j-1) \, \cdot \, p_1}{(4 Z)^\alpha} \enspace.
\]
For the number of nodes this yields
\[
\lvert K_Z(w) \rvert - 2 \leq 2 (j-1) \leq \frac{2}{\beta} \left( \frac{4 Z}{d(s_1, r_1)} \right)^\alpha \leq \frac{2}{\beta} \left( \frac{4 Z}{d} \right)^\alpha \enspace.
\]
\end{proof}

\begin{proof}[Proof of Proposition~\ref{prop:generalmetric:farlinks}]
Let $V(\L') = \{v_1, \ldots, v_{\bar{n}} \}$ such that $d(v_1, w) \leq d(v_2, w) \leq \ldots \leq d(v_{\bar{n}}, w)$. Note that for all $Z > 0$ we have $K_Z(w) = \{1, \ldots, x\}$ for some $x \in \NN$ by this definition.

For $j \leq \log \bar{n} + 1$ let $R_k = [2^k] \setminus [2^{k-1}] = \{2^{k-1}+1,\ldots,2^k\}$. Furthermore, let $Z_j$ be defined as $Z_j = \min_{i \in R_j} d(v_i, w) = d(v_{2^{j-1}+1},w)$.
These definitions yield:
\[
\sum_{i=1}^{\bar{n}} \frac{d^\alpha}{d(v_i, w)^\alpha} \leq \sum_{j=1}^{\log \bar{n} + 1} \sum_{i \in R_j} \frac{d^\alpha}{d(v_i, w)^\alpha} \leq d^\alpha \sum_{j=1}^{\log \bar{n} + 1} \frac{|R_j|}{Z_j^\alpha}
\]
As the distances are increasing, we have $Z_j \geq d(v_i, w)$ for all $i \leq 2^{j-1}$. In other words: $[2^{j-1}] \subseteq K_{Z_j}(w)$.

Now, we apply \lref[Fact]{fact:numberofnodesinneighborhood} on $|K_{Z_j}(w)|$, which gives
\[
2^{j-1} = |[2^{j-1}]| \leq |K_{Z_j}(w)| \leq \frac{2}{\beta} \left( \frac{4 Z_j}{d} \right)^\alpha + 2\enspace.
\]
Consequently, we have
\[
Z_j^\alpha \geq (2^{j-2}-1) \beta \left( \frac{d}{4} \right)^\alpha\enspace.
\]
Therefore, we can bound the sum by
\[
d^\alpha \sum_{j=1}^{\log \bar{n} + 1} \frac{2^{j-1}}{(2^{j-2}-1) \left( \frac{d}{4} \right)^\alpha} \leq \sum_{j=1}^{\log \bar{n} + 1} 4^{\alpha + 1} 
= \O( \log \lvert \L \rvert )
\]
This proves the claim.

\end{proof}
\end{appendix}

\end{document}